\documentclass[a4paper,11pt]{article}
\usepackage{amssymb,amsmath,latexsym,amsthm}
\usepackage{color,mathrsfs}
\usepackage{url}
\usepackage{enumerate}
\usepackage{dsfont}

\usepackage{endnotes}

\usepackage[dvips]{graphics}
\usepackage[xetex]{graphicx}
\usepackage{fancyhdr}
\usepackage{pstricks,pst-plot,pst-node,pstricks-add}

\title{Sufficient Condition for a Compact Local Minimality of a Lattice}

\author{Laurent B\'{e}termin\thanks{betermin@uni-heidelberg.de}\\ \\
Institut f\"{u}r Angewandte Mathematik, \\
Universit\"{a}t Heidelberg,\\
Im Neuenheimer Feld 294, \\
69120 Heidelberg. Deutschland }

\newtheorem{thm}{THEOREM}[section]
\newtheorem{defi}{Definition}[section]

\theoremstyle{definition}
\newtheorem{remark}[thm]{Remark}

\renewcommand{\textbf}[1]{\begingroup\bfseries\mathversion{bold}#1\endgroup}

\newcommand{\R}{\mathbb R}

\newcommand{\Z}{\mathbb Z}
\newcommand{\N}{\mathbb N}

\numberwithin{equation}{section}

\def\XXint#1#2#3{{\setbox0=\hbox{$#1{#2#3}{\int}$}
    \vcenter{\hbox{$#2#3$}}\kern-.5\wd0}}

\allowdisplaybreaks

\begin{document}
\maketitle
\begin{abstract} We give a sufficient condition on a family of radial parametrized long-range potentials for a compact local minimality of a given $d$-dimensional Bravais lattice for its total energy of interaction created by each potential. This work is widely inspired by the paper of F. Theil about two dimensional crystallization.
\end{abstract} 


\textbf{Mots-cl\'es} : Energie de r\'eseau ; Minimum local ; Potentiel ; R\`egle de Cauchy-Born ; Etat fondamental. \\

\section{Introduction}
As explained in \cite{Blanc:2015yu}, the crystallization problem, that is to say to understand why the particles structures are periodic at low temperature, is difficult and still open in most cases. Theil exhibited in \cite{Crystal} a radial parametrized long-range potential with the same form as the Lennard-Jones potential such that the triangular lattice is the ground state of the total energy in the sense of thermodynamic limit. This kind of potential, parametrized by a real number $\alpha>0$, is bigger than $\alpha^{-1}$ close to the origin, corresponding to Pauli's principle, has a well centred in $1$ and a $2\alpha$ width, its second derivative at $1$ is strictly positive and its decay  at infinity is $r\mapsto\alpha r^{-7}$. Thus, as small is $\alpha$, as close to $1$ is the mutual distance between nearest neighbours of the ground state configuration and as the interactions between distant points are negligible.\\ \\
In this paper, our idea is to present a family of parametrized potential very close to this one, with the most natural possible assumptions, such that a given Bravais lattice $L$ of $\R^d$ is a $N$-compact local minimum for the total energy of interaction. This kind of local minimality is called ``$N$-compact" because, given a maximal number $N$ of points that we want to move a little bit, there exists a maximal perturbation of the points which gives a larger total energy of interaction, in the sense that the difference of energies is positive. Moreover, as small is the parameter, as large the number $N$ can be chosen. We are strongly inspired by Theil's potential, keeping only local assumptions and strong parametrized decay. Furthermore, our work can be related to that of Torquato et al. about targeted self-assembly \cite{Torquato06,Torquato09} where they search for radial potentials such that a given configuration -- more precisely a part of a lattice -- is a ground state for the total energy of interaction.\\ \\
The aim of this paper is to give a generic construction of the family of potentials without taking into consideration the specific symmetries of $L$ or a particular pertubation of points.\\ \\
After defining the concepts and our parametrized potentials, we give the theorem, its proof and some important remarks and applications.

\section{Preliminaries : Bravais lattice and N-compact local minimality}
\begin{defi}
Let $d\in\N^*$, $(u_1,...,u_d)$ be a basis of $\R^d$ and $L=\bigoplus_{i=1}^n \Z u_i\subset \R^d$ be a Bravais lattice. For any $\lambda>0$, we define $m(\lambda):=\sharp \{L\cap \{\|x\|=\lambda\}  \}$ where $\|.\|$ denote the Euclidean norm and $\sharp A$ is the cardinal of set $A$. Moreover, we call $\lambda_1:=\min\{\|x\|;x\in L^*  \}$, where $L^*=L\backslash \{0\}$ and $\lambda_2:=\min\{\|x\|;\|x\|>\lambda_1,x\in L  \}$.\\
Furthermore, for a Bravais lattice $L\subset \R^d$ and $n>d$, we define the following lattice sums:
\begin{align*}
\zeta^*_L(n):=\sum_{x\in L \atop \|x\|>\lambda_1}\|x\|^{-n} \quad \textnormal{and } \quad \bar{\zeta}_L(n):=\sum_{x\in L \atop \|x\|>\lambda_1} (\|x\|-\lambda_1)^{-n}.
\end{align*}
\end{defi}

\begin{defi}
Let $L\subset \R^d$ be a Bravais lattice, $B\subset L$ a finite subset and $\alpha$ be a real number such that $0<\alpha<\lambda_1/2$. We say that $B^\alpha$ is an $\alpha$-compact perturbation of $B$ if
$$
\forall b\in B, \exists !b^{\alpha}\in B^{\alpha} \text{ such that } \left\|b-b^\alpha\right\|\leq \alpha.
$$
Moreover, if $B^\alpha$ is an $\alpha$-compact perturbation of $B\subset L$,  we write $L^\alpha(B):=(L\backslash B)\cup B^\alpha$ the perturbed lattice.
\end{defi}

\begin{defi}
Let $d\in \N^*$. We say that $V:\R_+^*\to \R$ is a \textbf{d-admissible potential} if $V$ is a $C^2$ function and, for any Bravais lattice $L\subset \R^d$,
$$
\sum_{x\in L^*}|V(\|x\|)| + \sum_{x\in L^*}\|x\||V'(\|x\|)|+\sum_{x\in L^*}\|x\|^2|V''(\|x\|)|<+\infty.
$$
\end{defi}
\begin{remark}
If, for any $k\in\{0,1,2\}$, $|V^{(k)}(r)|=O(r^{-p_k})$, $p_k>d+k$, then $V$ is $d$-admissible.
\end{remark}

\begin{defi}  Let $L$ be a Bravais lattice of $\R^d$ and $V$ a $d$-admissible potential. Let $N\in \N^*$, we say that $L$ is a \textbf{N-compact local minimum for the total $V$-energy} if for any subset $B\subset L$ such that $\sharp B\leq N$, there exists $\alpha_0>0$ such that for any $\alpha\in [0,\alpha_0)$  and any $\alpha$-compact perturbation $B^\alpha$ of $B$,
$$
\displaystyle \Delta_L^\alpha(V;B):=\sum_{b^\alpha\in B^\alpha}\sum_{y\in L^\alpha(B)\atop y\neq b^\alpha}V(\|b^\alpha-y\|)-\sum_{b\in B}\sum_{x\in L \atop x\neq b}V(\|b-x\|)\geq 0.
$$
\end{defi}

\section{Pressure, parametrized potentials and main result}
\begin{defi}
Let $V$ be a $d$-admissible potential and $L\subset \R^d$ be a Bravais lattice, then we define the \textbf{pressure of L submitted to $V$} by :
$$
\mathcal{P}(L,V):=-\sum_{x\in L^*}\|x\| V'(\|x\|).
$$
\end{defi}
\begin{remark}\label{rem}
Actually, as explained for instance in \cite{52871592}, if $\displaystyle E_{V}[L]:=\sum_{x\in L^*}V(\|x\|)$ is the energy per particle of $L$ submitted to $V$, i.e. the free energy at zero temperature, then, by usual thermodynamics formula, we define pressure $P(L,V)$ by
$$
P(L,V):=-\frac{dE_{V}[L]}{d A}=-\frac{1}{2A}\sum_{x\in L^*}\|x\| V '(\|x\|),
$$
where $A$ is the area per particle of $L$, that is to say $\mathcal{P}(L,V)=2A P(L,V)$.
\end{remark}

\begin{defi}\label{Pot} Let $L\subset \R^d$ be a Bravais lattice. We call \textbf{$(L,\theta)-$family} every set of $d$-admissible \textbf{$(L,\theta)-$potentials} $V_{\theta}:\R_+^*\to \R$, indexed by $\theta\in [0,\theta_0]$, $0\leq \theta_0<\lambda_1/2$, and satisfying the following conditions :
\begin{enumerate}
\item \textbf{Small pressure condition:} there exist positive real numbers $C,\mu$ such that for any $\theta$, 
$$
|\mathcal{P}(L,V_\theta)|\leq C\theta^{1+\mu};
$$
\item \textbf{Parametrized fast decay : }$\exists r_0\in[\lambda_1,\lambda_2)$, $\exists \varepsilon>0$, $\exists p>d$ such that for any $r>r_0$, 
$$
\displaystyle |V_\theta''(r)|\leq \theta^{2+\varepsilon}r^{-p-2};
$$
\item \textbf{Local convexity around first neighbours :} it holds $V_\theta''(r)\geq \eta>0$ on a neighborhood of the first distance $\lambda_1$, uniformly on $\theta$.
\end{enumerate}
\end{defi}
\vspace{4mm}
\begin{thm}
Let $L\subset \R^d$ be a Bravais lattice, then for any $N\in \N^*$, there exists $\theta_0>0$ such that for every $(L,\theta)-$family $(V_\theta)_{\theta\leq\theta_0}$ and every $\theta\in[0,\theta_0]$, $L$ is a N-compact local minimum for the total $V_\theta$-energy. Furthermore, in this case, the maximal perturbation $\alpha_0$ can be chosen equal to $\theta$.
\end{thm}

\begin{proof} Let $L$ be a Bravais lattice of $\R^d$. Let $N\in\N^*$ and $B:=\{b_1,...,b_N\}\subset L$. Let $\alpha_0$ be such that $0\leq\alpha_0< \lambda_1/2$ and $B^{\alpha_0}=\{b_1^{\alpha_0},...,b_N^{\alpha_0}\}$ be an $\alpha_0$-compact perturbation of $B$. For any $1\leq i \leq N$, for any $y\in L^{\alpha_0}(B)$, $y\neq b_i^{\alpha_0}$, and $x\in L$ such that $\|x-y\|\leq \alpha_0$, we define
$$
\alpha_{i,x}:=\|b_i^{\alpha_0}-y\|-\|b_i-x\|.
$$ 
Obviously, we have, for any $1\leq i\leq N$, for any $0\leq \alpha_0 <\lambda_1/2$ and any $x\in L$,
\begin{equation} \label{alphaestimate}
| \alpha_{i,x}|\leq 2\alpha_0.
\end{equation} 
 We assume, without loss of generality, that $\displaystyle \max_{i,x} |\alpha_{i,x}|=2\alpha_0$, left to decrease $\alpha_0$. We set $\theta\in[0,\lambda_1/2)$ and $V_\theta$ a $(L,\theta)-$potential. We have
$$
\Delta_L^{\alpha_0}(V_\theta;B)=\sum_{i=1}^N\sum_{y\in L^{\alpha_0}(B) \atop y\neq b_i^{\alpha_0}}V_\theta(\|b_i^{\alpha_0}-y\|)-\sum_{i=1}^N\sum_{x\in L \atop x\neq b_i}V_\theta(\|b_i-x\|).
$$
By Taylor expansion, we get, for any $1\leq i\leq N$, for any $y\in L^{\alpha_0}(B)$, $y\neq b_i^{\alpha_0}$, and $x\in L$ such that $\|x-y\|\leq \alpha_0$, 
$$
V_\theta(\|b_i^{\alpha_0}-y\|)\geq V_\theta(\|b_i-x\|)+\alpha_{i,x}V_\theta'(\|b_i-x\|)+\frac{\alpha_{i,x}^2}{2}V_\theta''(\xi_{i,x}),
$$
for suitable $\xi_{i,x}\in (\|b_i-x\|-|\alpha_{i,x}|,\|b_i-x\|+|\alpha_{i,x}|)$. Hence we obtain
$$
\Delta_L^{\alpha_0}(V_\theta;B)\geq \sum_{i=1}^N\sum_{x\in L \atop x\neq b_i}\alpha_{i,x}V_\theta'(\|b_i-x\|)+\frac{1}{2}\sum_{i=1}^N\sum_{x\in L \atop x\neq b_i}\alpha_{i,x}^2V_\theta''(\xi_{i,x}).
$$
We split interactions into two parts : the short range and the long range. For any $1\leq i\leq N$, we set
\vspace{-3mm}
\begin{align*}
\mathcal{S}_L^{i}:=\{ x\in L ; \|x-b_i\|=\lambda_1\}\quad \text{ and } \quad \mathcal{L}_L^{i}:=\{ x\in L ; \|x-b_i\|>\lambda_1 \}.
\end{align*}
Furthermore, as we assume that for all $r>r_0$ $\displaystyle |V_\theta''(r)|\leq \theta^{2+\varepsilon} r^{-p-2}$ and $V_\theta'$ goes to $0$ at $+\infty$, because it is a d-admissible potential, we have, by a simple argument, that  $\displaystyle |V_\theta'(r)|\leq \theta^{2+\varepsilon}r^{-p-1}$ for all $r>r_0$.\\
As $L$ is a Bravais lattice, it holds, for any $i$, $\displaystyle\sum_{x\in L\backslash\{ b_i\}} V_\theta'(\|b_i-x\|)\|b_i-x\|=-\mathcal{P}(L,V_\theta)$. Therefore,
\begin{align*}
\sum_{i=1}^N\sum_{x\in L \atop x\neq b_i}\alpha_{i,x}V_\theta'(\|b_i-x\|)=V_\theta '(\lambda_1)\left(\sum_{i=1}^N\sum_{x\in \mathcal{S}_L^{i}}\|b_i^{\alpha_0}-y\|\right)+ \sum_{i=1}^N\sum_{x\in \mathcal{L}_L^{i}}V_\theta'(\|b_i-x\|)\|b_i^{\alpha_0}-y\|+N\mathcal{P}(L,V_\theta).
\end{align*}
We remark that, writing $\displaystyle \Sigma_0:=\sum_{i=1}^N\sum_{x\in \mathcal{S}_L^{i}}\|b_i^{\alpha_0}-y\|$, and by definition of $\mathcal{P}(L,V_\theta)$,
\begin{align*}
V_\theta '(\lambda_1)\Sigma_0+N\mathcal{P}(L,V_\theta)= \mathcal{P}(L,V_\theta)\left(N-\Sigma_0(\lambda_1 m(\lambda_1))^{-1}  \right)-\left(\sum_{x\in L \atop \|x\|>\lambda_1}\|x\|V'_\theta(\|x\|)\right) \Sigma_0(\lambda_1 m(\lambda_1))^{-1}. 
\end{align*}
By \eqref{alphaestimate}, we have, for any $1\leq i \leq N$, for any $x\in \mathcal{S}_L^i$, $\lambda_1-2\alpha_0 \leq \|b_i^{\alpha_0} -y\|\leq \lambda_1+2\alpha_0$, and we get
$$
m(\lambda_1)N(\lambda_1-2\alpha_0)\leq \Sigma_0\leq m(\lambda_1)N(\lambda_1+2\alpha_0).
$$
Thus, we show that $\displaystyle \mathcal{P}(L,V_\theta)\left(N-\Sigma_0(\lambda_1 m(\lambda_1))^{-1}  \right)\geq -\frac{2 |\mathcal{P}(L,V_\theta)|N\alpha_0}{\lambda_1}\geq -\frac{2C}{\lambda_1}N\theta^{1+\mu}\alpha_0$, by assumption. Therefore, we obtain
\begin{align*}
V_\theta '(\lambda_1)\Sigma_0+N\mathcal{P}(L,V_\theta)\geq -\frac{2C}{\lambda_1}N\theta^{1+\mu}\alpha_0 - \theta^{2+\varepsilon}\zeta^*_L(p)N\left(1+\frac{2\alpha_0}{\lambda_1}  \right).
\end{align*}
Hence, for first order terms, we get
\begin{align*}
\sum_{i=1}^N\sum_{x\in L \atop x\neq b_i}\alpha_{i,x}V_\theta'(\|b_i-x\|)\geq -\frac{2C}{\lambda_1}N\theta^{1+\mu}\alpha_0-2\zeta^*_L(p)N\theta^{2+\varepsilon}-2\left( \frac{\zeta^*_L(p)}{\lambda_1}+\zeta^*_L(p+1) \right)N\theta^{2+\varepsilon}\alpha_0.
\end{align*}
For the second order terms, as  $\displaystyle \max_{i,x} |\alpha_{i,x}|=2\alpha_0$, we have $\displaystyle \sum_{i=1}^N\sum_{x\in\mathcal{S}_L^{i}}\alpha_{i,x}^2 \geq 4\alpha_0^2$, and we obtain
\begin{align*}
\frac{1}{2}\sum_{i=1}^N\sum_{x\in L \atop x\neq b_i}\alpha_{i,x}^2V_\theta''(\xi_{i,x})\geq 2\eta \alpha_0^2 -2\theta^{2+\varepsilon}\alpha_0^2 N\bar{\zeta}_L(p+2).
\end{align*}
Finally we get, for any $0\leq \alpha_0< \lambda_1/2$, 
\begin{align*}
\Delta_L^{\alpha_0}(V_\theta;B)\geq & 2\eta \alpha_0^2-2\theta^{2+\varepsilon}\alpha_0^2 N\bar{\zeta}_L(p+2)-\frac{2C}{\lambda_1}N\theta^{1+\mu}\alpha_0-2\zeta^*_L(p)N\theta^{2+\varepsilon}\\
&-2\left( \frac{\zeta^*_L(p)}{\lambda_1}+\zeta^*_L(p+1) \right)N\theta^{2+\varepsilon}\alpha_0\\
&=2\eta \alpha_0^2 - N\left(A\theta^{2+\varepsilon}+B\theta^{1+\mu}\alpha_0+C\theta^{2+\varepsilon}\alpha_0+D\theta^{2+\varepsilon}\alpha_0^2  \right),
\end{align*}
where positive real numbers $A,B,C,D$ depend only on $L$.
Given $\theta\in [0,\lambda_1/2)$, if $\alpha_0=\theta$, then
\begin{equation}\label{Delta}
\Delta_L^{\theta}(V_\theta;B)\geq  2\eta\theta^2-N\left( A\theta^{2+\varepsilon}+B\theta^{2+\mu}+ C\theta^{3+\varepsilon}+D\theta^{4+\varepsilon} \right).
\end{equation}
As $\eta>0$, there exists $\theta_0 \in[0,\lambda_1/2)$, depending on $N$, sufficiently small such that for any $\theta\in [0,\theta_0]$ and for any $\alpha\in [0,\theta]$, $\Delta_L^\alpha(V_\theta;B)\geq 0$ and then $L$ is a $N$-compact local minimum for the total $V_\theta$-energy, for any $(L,\theta)-$potential of the $(L,\theta)-$family $(V_\theta)_{\theta\leq \theta_0}$.
\end{proof}

\section{Remarks}
\textbf{1. Isothermal compressibility.} It is usual (see \cite{52871592,14163101}) to define the isothermal compressibility $\kappa_T$, from pressure $P$ (see Remark \ref{rem}), by
\begin{align*}
\frac{1}{\kappa_T}:=-A\frac{dP}{dA}=A\frac{d^2E_{V_\theta}[L]}{dA^2}=\frac{1}{4A}\sum_{x\in L^*}\left[\|x\|^2 V_{\theta}''(\|x\|)-\|x\|V_\theta '(\|x\|)  \right],
\end{align*}
where $A$ is the area per particle of $L$ and $E_{V_\theta}[L]$ its energy per point. We know that $\kappa_T> 0$ (see \cite[Section 5.1]{14163101}). Actually, that follows here from assumptions on $V_\theta$, if $\theta$ is sufficiently small. Indeed, by assumptions, we have, for $\theta$ sufficiently small,
$$
\frac{4A}{\kappa_T}=\sum_{x\in L^*}\left[\|x\|^2 V_{\theta}''(\|x\|)-\|x\|V_\theta '(\|x\|)  \right]\geq \lambda_1^2 m(\lambda_1)\eta-\theta^{2+\varepsilon}\zeta^*_L(p)-C\theta^{1+\mu}>0.
$$
\textbf{2. Zero pressure condition and local minimality among dilated of $L$.} Let us assume here that, for any $\theta\in [0,\lambda_1/2)$,
\begin{equation}\label{zeropressure}
\mathcal{P}(L,\theta):=-\sum_{x\in L^*} \|x\| V_\theta '(\|x\|)=0,
\end{equation}
which is thermodynamically natural at zero temperature, for instance if $L$ is the cooling of an ideal gas. 
Now if we consider $f:r\mapsto E_{V_\theta}[rL]$, we get, by \eqref{zeropressure} and $\kappa_T> 0$ (see previous remark), $f'(1)=0$ and $f''(1)> 0$, i.e. \textbf{$L$ is a local minimum of $L\mapsto E_{V_\theta}[L]$ among its dilated}, which seems natural if $L$ is a $N$-compact local minimum for the total $V_\theta$-energy for $N$ arbitrarily large, and this is actually assumed in Theil's paper \cite{Crystal}. \\ \\
However, the reverse is false. A Bravais lattice can be a local minimum among its dilated for the energy per point but not a $N$-compact local minimum for the total energy. For instance, if $d=1$, $L=\Z$, $N=1$ and $V$ defined by $V(r)=0$ for $r\geq 5/2$, $V'(1)=V'(2)=0$, $V''(1)=-1$ and $V''(2)=1/3$. We have $\displaystyle\sum_{x\in \Z^*}|x|V'(|x|)=0$ and $\displaystyle\sum_{x\in\Z^*}|x|^2V''(|x|)=2/3\geq 0$, hence $\Z$ is a local minimum among lattices of the $V$-energy per point. For $\alpha\geq 0$, we estimate, by Taylor expansion,
\begin{align*}
\Delta^{\alpha}(V;L)&=\sum_{x\in\Z^*}\left[V(|x-\alpha|)-V(|x|)  \right]\\
&=\alpha^2V''(1)+\alpha^2V''(2)+\alpha^2\phi(\alpha)=\alpha^2(-2/3+\phi(\alpha)),
\end{align*}
where $\phi(\alpha)$ goes to $0$ as $\alpha\to 0$. Hence for $\alpha<\alpha_0$ sufficiently small, $-2/3+\phi(\alpha)<0$ and $\Z$ is not a $1$-compact local minimum of the total $V$-energy.\\ \\
\textbf{3. Effects of parameters $\varepsilon,\mu,p$ and $\eta$.} By \eqref{Delta}, our assumptions on $V_\theta$ give indications about the stability of lattice $L$:
\begin{itemize}
\item Range : a large $p$, $C$, $\varepsilon$ or $\mu$ allow to take a large perturbation $\alpha_0$ for fixed $N$, i.e. a better decay at infinity implies a stronger stability of the lattice;
\item Second derivative around nearest-neighbours distance : a large $\eta$ also allows a large perturbation $\alpha_0$ for fixed $N$. Typically, a narrow well around $\lambda_1$ ``catches" the first neighbours of the minimizing configuration at distance $\lambda_1$.
\end{itemize}
\textbf{4. Difference between the decay after the first distance and the perturbation.} We can see that $\theta_0^{2+\varepsilon}<<\theta_0$, i.e. the maximum decay is really smaller than the maximum perturbation and this allows not to assume a local behaviour of $V_\theta$ around $\lambda_1$ with respect to $\theta$, as in Theil's work. Obviously, if $\theta=0$ then $V_0(r)=0$ for any $r>r_0$ and $V_0'(\lambda_1)=0$, therefore $\lambda_1$ is a local minimum of $V_0$ and the potential is short-range : only the first neighbours interact and the $N$-compact local minimality is clear for any $N$ with a perturbation $\alpha_0$ as small as $N$ is large. \\ \\
\textbf{5. A kind of Cauchy-Born rule.} Our result can be viewed like a justification of a kind of Cauchy-Born rule (see \cite{TheilOrtner,WeinanEPMing}). Indeed, if we consider a solid as a Bravais lattice $L$ where the inside is a finite part of $L$ with cardinal $N$ and the rest is its boundary, a small perturbation of the inside, depending on $N$, increases the total energy of interaction in the solid. That is to say that the inside of the solid follows its boundary to a stable configuration.\\ \\
\noindent \textbf{Acknowledgements:} I am grateful to Florian Theil, Xavier Blanc and Etienne Sandier for their interest and helpful discussions. I wish to express my thanks to the anonymous referee for her/his suggestions.

\bibliographystyle{plain}
\bibliography{biblio}

\end{document}